\documentclass[a4paper,reqno,oneside]{amsart}
\usepackage{enumitem,dsfont,amssymb}
\usepackage[colorlinks=true, linkcolor=blue, citecolor=red, urlcolor=magenta,hyperfootnotes=false]{hyperref}
\usepackage{cleveref}
\usepackage{graphicx}
\usepackage{upgreek}

\newtheorem{thm}{Theorem}
\newtheorem{defi}{Definition}
\newtheorem{lem}[thm]{Lemma}
\newtheorem{prop}[thm]{Proposition}

\theoremstyle{definition}

\newtheorem{rem}[thm]{Remark}
\newtheorem{oppb}[thm]{Open problem}

\newcommand{\CC}{\mathbb{C}}
\newcommand{\RR}{\mathbb{R}}

\newcommand{\NN}{\mathbb{N}}
\newcommand{\ZZ}{\mathbb{Z}}
\newcommand{\cK}{\mathcal{K}}

\newcommand{\cD}{\mathcal{D}}

\newcommand{\Rt}{\RR^3}
\newcommand{\Vep}{\mathbf{V}_{\!\epsilon}}

\newcommand{\sign}{\mathop{\textrm{sign}}}

\DeclareMathOperator{\ran}{ran}

\title[Relativistic $\delta$-shell interactions]{Dirac operators and shell interactions: a survey}

\author{Thomas Ourmi\`eres-Bonafos}
\address{CNRS \& CEREMADE, Universit\'e Paris-Dauphine, PSL Research University, F-75016 Paris, France}

\email{ourmieres-bonafos@ceremade.dauphine.fr}
\urladdr{http://www.ceremade.dauphine.fr/~ourmieres/}

\author{Fabio Pizzichillo}
\address{CNRS \& CEREMADE, Universit\'e Paris-Dauphine, PSL Research University, F-75016 Paris, France}

\email{pizzichillo@ceremade.dauphine.fr}
\urladdr{http://www.ceremade.dauphine.fr/~pizzichillo/}
\begin{document}


\begin{abstract}In this survey we gather recent results on Dirac operators coupled with $\delta$-shell interactions. We start by discussing recent advances regarding the question of self-adjointness for these operators. Afterward we switch to an approximation question: can these operators be recovered as limits of Dirac operators coupled with squeezing potentials ?

We also discuss spectral features of these models. Namely, we recall the main spectral consequences of a resolvent formula and conclude the survey by commenting a result of asymptotic nature for the eigenvalues in the gap of a Dirac operator coupled with a Lorentz-scalar interaction.
\end{abstract}                                                            

\maketitle

\section{Introduction}
\subsection{Singular interactions in non-relativistic quantum mechanics}\label{subsec:non-relqm}
Some non-relativistic quantum systems are efficiently described by Schr\"odinger operators with singular $\delta$-type potentials supported on a zero Lebesgue measure set.

For example, such hamiltonians arise as approximations of atomic systems in strong homogeneous magnetic fields \cite{BD06} or when investigating photonic crystals with high contrast \cite{FK96}.

In this survey, we focus on the particular case of a bounded lipschitz surface  without boundary $\Sigma\subset \mathbb{R}^3$ which splits the euclidean space $\mathbb{R}^3$ into two domains $\Omega_\pm \subset \mathbb{R}^3$
\[
	\mathbb{R}^3 = \Omega_+ \cup \Omega_- \cup \Sigma.
\]
Such a surface $\Sigma$ is called a {\it shell} and we consider a hamiltonian acting in $L^2(\mathbb{R}^3)$ which formally writes
\begin{equation}\label{eqn:defschrod}
	S_\tau:= - \Delta - \tau \delta_\Sigma,
\end{equation}
where $\tau \in \mathbb{R}$ is a coupling parameter and $\delta_\Sigma$ is the distribution defined for all $\varphi \in C_0^\infty(\mathbb{R}^3)$ as
\begin{equation}\label{eqn:defdelta}
	\langle\delta_\Sigma,\varphi\rangle_{\mathcal{D}(\mathbb{R}^3),\mathcal{D}'(\mathbb{R}^3)} : = \int_\Sigma \varphi\, ds.
\end{equation}
Here, $ds$ denotes the two-dimensional Haussdorff measure on $\Sigma$.\\

\subsection*{Definition and self-adjointness} In order to investigate rigorously the operator $S_\alpha$ given in \eqref{eqn:defschrod}, one has to answer first the following two preliminary questions.
\begin{itemize}
	\item[(Q1):]\label{enum:1} How can the Schr\"odinger operator $S_\tau$ be defined rigorously ?
	\item[(Q2):]\label{enum:2} Is the Schr\"odinger operator $S_\tau$ self-adjoint ?
\end{itemize}
Both these questions are answered in \cite{BEL14} to which we refer for a rigorous and detailed approach. Nevertheless, for further purpose, we recall here the usual strategy. Start by considering the bilinear form
\[
	s_\tau [u,v] := \int_{\mathbb{R}^3} \langle \nabla u,\nabla v\rangle_{\mathbb{R}^3} dx - \tau \int_{\Sigma}u\bar{v}\, ds,\quad u,v\in H^1(\mathbb{R}^3).
\]
It is well known that this bilinear form is symmetric, densely defined, closed and semi-bounded below in $L^2(\mathbb{R}^3)$ (see, for instance, \cite[\S 2.]{BEKS94}). In particular, the Schr\"odinger operator $S_\tau$ can be properly defined as the self-adjoint operator associated to the bilinear form $s_\tau$ thanks to Kato's first representation theorem \cite[Ch. VI, Thm. 2.1]{Kat95}. In particular, it implies $\mathcal{D}(S_\tau)\subset H^1(\mathbb{R}^3)$.

Nevertheless, one could argue that such an implicit definition of $\mathcal{D}(S_\tau)$ does not fully answers question (Q1) in the sense that we do not know neither the action of the Schr\"odinger operator $S_\tau$ nor have described its domain. Actually, it can be proved that
\begin{equation}
	\left\{\begin{array}{rcl}\mathcal{D}(S_\tau)  &=& \{ u = u_+\oplus u_- \in \mathcal{D}(s_\tau) : \Delta u_\pm \in L^2(\Omega_\pm),\\&&\quad\quad\quad\quad\quad\quad\quad u_+|_\Sigma = u_-|_\Sigma,\quad \partial_n u|_\Sigma = \tau u|_\Sigma\},\\S_\tau(u_+\oplus u_-) & = & (-\Delta u_+) \oplus (-\Delta u_-),
		\end{array}\right.
		\label{eqn:defschrodrigo}
\end{equation}
where we have identified $L^2(\mathbb{R}^3)$ with $L^2(\Omega_+) \oplus L^2(\Omega_-)$ and where $\partial_n$ denotes the usual jump of the normal derivatives of $u_+$ and $u_-$ through the surface $\Sigma$.

Here, of course, as $\mathcal{D}(S_\tau) \subset H^1(\mathbb{R}^3)$ one can easily give a sense to the traces $u_\pm|_\Sigma \in H^{\frac12}(\Sigma)$ (see, for instance, \cite[Thm. 3.37]{Mcl00}). However the jump of the traces of the normal derivatives has to be understood in a weak sense, which is possible because $\Delta u_\pm \in L^2(\Omega_\pm)$.

Such jump conditions can be recovered using a naive approach. Indeed, take for instance $u \in C^\infty(\mathbb{R}^3\setminus\Sigma) \cap L^2(\RR^3)$ and apply the expression of the Schr\"odinger operator $S_\tau$ given in \eqref{eqn:defschrod} in the sense of distributions. Using the jump formula, it gives:
\begin{align}
	(-\Delta - \tau \delta_\Sigma ) u &= - \Big(\sum_{j=1}^3 (\partial_j^2u_+)\mathds{1}_{\Omega_+} + (\partial_j^2u_-)\mathds{1}_{\Omega_-}\Big)\nonumber\\&\quad + (\partial_{n_-} u_-|_\Sigma +  \partial_{n_+} u_+|_\Sigma) \delta_\Sigma - \tau u|_\Sigma \delta_\Sigma,
	\label{eqn:jumnrq}
\end{align}
where $\partial_{n_\pm}u_\pm|_\Sigma$ are the Neumann traces of $u_\pm$. For the right-hand side to belong to $L^2(\mathbb{R}^3)$ we need the following equality to hold:
\[
	\partial_n u|_\Sigma := \partial_{n_-} u_-|_\Sigma +  \partial_{n_+} u_+|_\Sigma = \tau u|_\Sigma.
\]
It is exactly the jump condition given in \eqref{eqn:defschrodrigo}. In particular the operator $S_\tau$ defined in \eqref{eqn:defschrodrigo} acts as expected in \eqref{eqn:defschrod}.\\
\paragraph{\bf Approximations of $\delta$-shell potentials}
\label{sub:approx-Schr}
From a physical point of view, a Schr\"odinger operator with a $\delta$-shell potential is an idealized hamiltonian for a quantum particle submitted to electric potential localized in a thin tubular neighborhood of the shell.

To justify this modeling, pick a function $V\in C^\infty_0(\RR^3)$ and let $V_\epsilon$ be a sequence of mollifiers such that
\[
V_\epsilon\to \Big(\int_{\RR^3} V(x) dx \Big) \delta_\Sigma,\quad\epsilon\to0.
\]
One can wonder if the family of hamiltonians $(-\Delta - V_\epsilon)_{\epsilon>0}$ has a limit when $\epsilon\to0$. These operators are self-adjoint on the domain $H^2(\RR^3)$ and for the model to be physically consistant, we would like to obtain a connection between the spectrum $Sp(-\Delta - V_\epsilon)$ of $-\Delta - V_\epsilon$ and the spectrum of its limit operator. Mathematically, this can be answered investigating the following question (see \cite[Section VIII.7]{RS1}).
\begin{itemize}
\item[(Q3):]\label{item:Q3}
For some operator topology, does the following convergence hold
\[
-\Delta-V_\epsilon\underset{\epsilon\to0}{\longrightarrow}
S_{\tau_V},
\]
for some $\tau_V \in \RR$ depending on the potential $V$ ?
\end{itemize}

Finally, one would like to know whether or not this limiting procedure allows to recover the whole range of all possible coupling constants:
\begin{itemize}
\item[(Q4):] Given an operator $S_\tau$, can it be realized as an operator $S_{\tau_V}$, for some $V\in C_0^\infty(\RR^3)$ ?
\end{itemize}

Both Questions (Q3) and (Q4) are well investigated in the literature. Let us mention the one-dimensional case studied in \cite
{AGHH12} and \cite{BEHL2017} where the case of higher dimensions is dealt with for singular perturbations on general smooth hyper-surfaces. This question is also discussed in \cite[\S 10.1]{EH15}.

The main result is a \emph{norm resolvent} convergence of the family of operators $(-\Delta - V_\epsilon)_{\epsilon>0}$ to $S_{\tau_V}$ where $\tau_V = \int_{\RR^3}V(x) dx$, answering both questions (Q3) and (Q4).\\
\paragraph{\bf Spectral theory}The structure of the spectrum of the operator $S_\tau$ attracted a lot of attention and is well understood, in particular as $\Sigma$ is compact we get
\[
	Sp_{ess}(S_\tau) = [0,+\infty),
\]
see \cite[Theorem 2.1]{BEL13} and if $\tau>0$, the interaction becomes attractive and bound states can appear below the threshold of the essential spectrum. The existence of such bound states, as well as their behavior in the strong coupling regime $\tau\to +\infty$ has been intensively investigated and we refer to \cite[Chapter 10]{EH15} and references therein for results in this direction.

\subsection{Singular interactions in relativistic quantum mechanics}
\label{sec:dirac}
The aim of this survey is to know until what extent questions (Q1)-(Q4) have been investigated for relativistic quantum particles. In this case, the Schr\"odinger operator \eqref{eqn:defschrod} is replaced by the Dirac operator that acts in $L^2(\mathbb{R}^3,\mathbb{C}^4)$ as
\begin{equation}\label{eq:def-free-dirac}
D_m:= D :=  -i \Big(\sum_{j=1}^3\alpha_j\partial_j\Big) + m\beta=-i \alpha\cdot\nabla+ m\beta 
\end{equation}
where $m\in \mathbb{R}$ is the mass of the considered particle and $\alpha_1,\alpha_2,\alpha_3,\beta \in \mathbb{C}^{4\times4}$ are the Dirac matrices
\[
	\alpha_j := \begin{pmatrix}0 & \sigma_j \\ \sigma_j & 0\end{pmatrix}, \quad \beta=\begin{pmatrix}1_2 & 0\\0 & -1_2 \end{pmatrix}.
\]
Here $\sigma_1,\sigma_2,\sigma_3 \in \mathbb{C}^{2\times2}$ are the usual Pauli matrices
\[
	\sigma_1 := \begin{pmatrix} 0&1\\1&0\end{pmatrix},\quad \sigma_2 := \begin{pmatrix} 0&-i\\i&0\end{pmatrix}, \sigma_3:= \begin{pmatrix} 1&0\\0&-1\end{pmatrix}.
\]
It is well known (see \cite[\S 1.4]{Tha92}) that $D$ is essentially self-adjoint on $C_0^{\infty}(\Rt,\CC^4)$. We denote by $D_{\rm free}$ its self-adjoint extension and $D_{\rm free}$ is called the free Dirac operator. It is defined on the domain $\cD(D_{\rm free}):=H^1(\Rt,\CC^4)$ and
\begin{equation}\label{eq:spec-free-dirac}
\operatorname{Sp}(D_{\rm free})=\operatorname{Sp}_{ess}(D_{\rm free})=(-\infty,-|m|]\cup[|m|,+\infty).
\end{equation}

The aim of this review is to gather recent advances in the study of the Dirac operator coupled with a $\delta$-shell interaction. This operator formally acts as
\begin{equation}\label{eqn:defdirac}
	D_{\tau,\eta} := D + (\tau 1_4 + \eta \beta)\delta_\Sigma,
%	\\:=-i \alpha\cdot\nabla+ m\beta + (\tau 1_4 + \eta \beta)\delta_\Sigma,
\end{equation}
where $\tau,\eta\in\mathbb{R}$ are coupling constants.

From a physical point of view, this operator arises when one aims to study relativistic properties of spin-$1/2$ particles (such as electrons) coupled with an electrostatic potential of interaction strength $\tau$ and a Lorentz-scalar potential of interaction strength $\eta$, both localized on a shell $\Sigma$.

Now, one could ask whether the relativistic counterpart of question (Q1) has a natural answer as in the non-relativistic case evoked in \S \ref{subsec:non-relqm} and it turns out the question is actually more involved, due to the relativistic nature of the model. Indeed, the energy functional of the Dirac operator is neither bounded below nor above as it can be seen by looking at the spectrum of the free-Dirac operator \eqref{eq:spec-free-dirac}.

Consequently, the approach involving a quadratic form is not available anymore and one need to think about another rigorous strategy.

When trying to apply the program (Q1)-(Q4) to relativistic particles we can see that the answer is not as straightforward as in the non-relativistic setting and this review aims to illustrate the state of the art regarding these questions.
	
\subsection{Structure of the survey}
Section \ref{sec:defsa} is devoted to the rigorous definition of the Dirac operator coupled with both electrostatic and scalar interactions supported on a shell $\Sigma$, answering to questions (Q1) and (Q2) in the relativistic setting.

Section \ref{sec:approxproc} aims to justify that the the Dirac operator coupled with either an electrostatic or a scalar interaction supported on a shell $\Sigma$ can be approached by a sequence of squeezing potentials, answering to questions (Q3) and (Q4).

Finally, Section \ref{sec:strucspec} deals with various properties of the spectrum of this operator that can be deduced from the previous definitions and results of Section \ref{sec:defsa}. Namely, a resolvent formula is given and spectral asymptotics are obtained in the large mass limit for a pure Lorentz-scalar potential.

\section{Definition of relativistic shell interactions and self-adjointness}\label{sec:defsa}
In this section we discuss the various approach used in the past few years to define the Dirac operator with a shell interaction. \S \ref{subsec:DES} - \ref{subsec:critstreng} follow the chronological order of publications in order to emphasize on the key evolutions. Namely, we discuss the question of self-adjointness as dealt with in \cite{DES89,AMV14,BEHL18,OBV18,BH17}.

A reader only interested in the present state of the art can skip directly to \S \ref{subsec:stofart} where we sum up the main results. We also state two open problems related to the question of self-adjointness of the Dirac operator with a shell interaction.
%%%%%%%%%%%%%%%%%%%%%
%%%%%%%%%%%%%%%%%%%%%
\subsection{The spherically symmetric \texorpdfstring{$\delta$}{\$delta\$}-shell}\label{subsec:DES}
%%%%%%%%%%%%%%%%%%%%%
%%%%%%%%%%%%%%%%%%%%%
The first definition of the operator $D_{\tau,\eta}$ is given in \cite{DES89} where the special case $\Sigma= \mathbb{S}^2$ is considered. The authors look for a definition of the operator $D_{\tau,\eta}$ which preserves the spherical symmetry of the problem.
To do this, they decompose the ambient Hilbert space $L^2(\mathbb{R}^3,\mathbb{C}^4)$ in \emph{partial wave subspaces} associated to the Dirac operator. Namely, they are reduced to investigate the self-adjoint extensions of a countable family of operators of the form
\[
\left\{\begin{array}{rcl}
	\mathcal{D}(d)& = & C_0^\infty\big((0,R)\cup(R,+\infty),\mathbb{C}^2\big),\\
	d u &=&\Big(-i\sigma_2 \frac{d}{dr} + m\sigma_3 + \frac{\chi}{r}\sigma_1\Big)u,
	\end{array}\right.
\]
where these operators act in $L^2\big((0,+\infty),\mathbb{C}^2\big)$ and where $\chi\in\mathbb{Z}\setminus\{0\}$ (see \cite[\S.3]{DES89} for details).

Then, they suggest an extension $d_{\tau,\eta}$ of $d$ defined as
\[
\left\{\begin{array}{rcl}
	\mathcal{D}(d_{\tau,\eta})& = & \bigg\{u = u_+\oplus u_- \in AC(I_-,\mathbb{C}^2)\oplus \big(AC(I_+,\mathbb{C}^2)\cap L^2(I_+,\mathbb{C}^2)\big):\\&&\quad\quad\quad\quad \Big(-i\sigma_2 \frac{d}{dr} + \frac{\chi}{r}\sigma_1\Big)u_\pm \in L^2(I_\pm,\mathbb{C}^2),\\&&\quad\quad\quad\quad\quad\frac12\big(\tau 1_2 + \eta \sigma_3\big)\big(u_+(1) + u_-(1)\big) = i\sigma_2\big(u_+(1) - u_-(1)\big)\bigg\}\\
	d_{\tau,\eta} u &=&\Big(-i\sigma_2 \frac{d}{dr} + m\sigma_3 + \frac{\chi}{r}\sigma_1\Big)u_+ \oplus \Big(-i\sigma_2 \frac{d}{dr} + m\sigma_3 + \frac{\chi}{r}\sigma_1\Big)u_-,
	\end{array}\right.	
\]
where we have set $I_+ = (0,1), I_- = (1,+\infty)$.

Such a choice for the jump condition at $r=1$ is justified as follows. The distribution $\delta_{\mathbb{S}^2}$ on the shell $\mathbb{S}^2$ can be understood in the partial wave decomposition as $\delta_{\{r = 1\}}$, the  Dirac distribution at $r=1$.

However, for a function $u \in AC(I_-)\oplus AC(I_+)$, the expression $(u\delta_{\{r=1\}})$ does not make any sense {\it a priori}. Hence, in \cite[Eqn. (4.4)]{DES89} they {\it choose} to define a distribution $(u\delta_{\{r=1\}})$ as:
\begin{equation}\label{eqn:firstsymshell}
	u\delta_{\{r=1\}} = \frac12\big(u_+(1) + u_-(1)\big)\delta_{\{r=1\}}.
\end{equation}
We emphasize on the fact that this is a choice. The $\delta$-shell is said to be symmetric because each boundary term $u_\pm(1)$ is considered with a coefficient $\frac12$ and although they may not be physically interesting, asymmetric $\delta$-shell can also be considered (see \cite[Appendix]{DES89}).

Thus, the strategy of Dittricht, Exner and $\check{\text{S}}$eba gives a natural answer to (Q1) in the relativistic setting. The operator $D_{\tau,\eta}$ is defined {\it via} a fiber decomposition in partial wave subspace and question (Q2) about self-adjointness is answered thanks to the following proposition.
\begin{prop}[{\cite[Prop. 4.1]{DES89}}] The operator $d_{\tau,\eta}$ is self-adjoint.
\end{prop}
Remark that this strategy defines the domain $\mathcal{D}(D_{\tau,\eta})$ as the direct sum of the domains of the (countable) partial wave operators $\mathcal{D}(d_{\tau,\eta})$ and every information on the Sobolev regularity of functions in $\mathcal{D}(D_{\tau,\eta})$ is implicitly hidden in this decomposition.
%%%%%%%%%%%%%%%%%%%%%
%%%%%%%%%%%%%%%%%%%%%
\subsection{General shells}\label{subsec:genshell}
%%%%%%%%%%%%%%%%%%%%%
%%%%%%%%%%%%%%%%%%%%%
Let us describe the approach of \cite{AMV14}, where the authors study the case of a general Lipschitz shell $\Sigma$. Consider the minimal Dirac operator
\[
	\left\{\begin{array}{rcl}\mathcal{D}(D_{\min}) &=& C_0^\infty(\mathbb{R}^3\setminus\Sigma,\mathbb{C}^4),\\ D_{\min} u &=& (-i\alpha\nabla + m\beta)u.
	\end{array}\right.
\]
This operator is symmetric and we consider its adjoint, the maximal operator $D_{\max}$, defined as
\[
	\left\{\begin{array}{rcl}\mathcal{D}(D_{\max}) &=& \{u \in L^2(\mathbb{R}^3,\mathbb{C}^4) : (\alpha\cdot\nabla)(u_\pm) \in L^2(\Omega_\pm,\mathbb{C}^4)\}\\ D_{\max} u &=& \big((-i\alpha\cdot\nabla + m\beta)u_+\big)\oplus\big((-i\alpha\cdot\nabla + m\beta)u_-\big),
	\end{array}\right.
\]
where once again we have set $u_\pm := u\mathds{1}_{\Omega_\pm}$ and identified $L^2(\Omega_+,\mathbb{C}^4) \oplus L^2(\Omega_-,\mathbb{C}^4)$ with $L^2(\mathbb{R}^3,\mathbb{C}^4)$.

Their main idea is to define the operator $D_{\tau,\eta}$ on a subdomain of $\mathcal{D}(D_{\max})$ providing extra jump conditions through the shell $\Sigma$.\\

\paragraph{\underline{First step}}
They try to give an accurate description of the domain $\mathcal{D}(D_{\max})$ using various integral operators involving a fundamental solution $\phi$ of the free Dirac operator $D_{\rm free}$.
\begin{prop}[{\cite[Lemma 3.1]{AMV14}}] Let $m>0$. A fundamental solution of the free Dirac operator $D_{\rm free}$ is given by
\[
	\phi(x) = \frac{e^{-mx|}}{4\pi|x|}\big(m\beta + (1+m|x|)i\alpha\cdot\frac{x}{|x|^2}\big),\quad \text{for } x\in\mathbb{R}^3\setminus\{0\}.
\]
\label{prop:def-fundsol}
\end{prop}
Then, they construct a linear and bounded operator $\Phi : L^2(\Sigma,\mathbb{C}^4)\to L^2(\mathbb{R}^3,\mathbb{C}^4)$ defined as
\begin{equation}\label{eqn:bdope}
	\Phi(g)(x) := \int_\Sigma \phi(x-y) g(y) ds(y),\quad\text{for } x\in\mathbb{R}^3\setminus \Sigma,
\end{equation}
see \cite[Corollary 2.3]{AMV14}.

Remark that the operator $\Phi$ is constructed in order to have for all $g\in L^2(\Sigma,\CC^4)$ $D\Phi(g) =0$ as a distribution in $\mathcal{D}'(\Omega_\pm)$. In particular, $\Phi(g)$ is harmonic for the Dirac operator in the domains $\Omega_\pm$ and $\Phi(g)\in\mathcal{D}(D_{\max})$. Then, instead of working with the space $\mathcal{D}(D_{\max})$, they focus on its subspace $E$ defined as
\begin{equation}\label{eqn:maxdomain}
	E := \{ u + \Phi(g) : u\in H^{1}(\mathbb{R}^3,\mathbb{C}^3), g\in L^2(\Sigma,\mathbb{C}^4)\} \subset \mathcal{D}(D_{\max}).
\end{equation}
\paragraph{\underline{Second step}} They prove that functions in the vectorial space $E$ have non-tangential traces in $L^2(\Sigma,\mathbb{C}^4)$. More precisely, one can define two linear bounded operators $C_\pm: L^2(\Sigma,\mathbb{C}^4)\to L^2(\Sigma,\mathbb{C}^4)$ defined as
\[
	C_\pm(g)(x) := \lim_{\Omega_\pm \ni y \overset{nt}{\to} x} \Phi(g)(y),
\]
and they are related {\it via} a Plemelj-Sokhotski jump formula to the linear and bounded operator $C_s : L^2(\Sigma,\mathbb{C}^4)\to L^2(\Sigma,\mathbb{C}^4)$ defined for $x\in L^2(\Sigma,\mathbb{C}^4)$ as
\begin{equation}\label{eqn:PSjump}
	C_s(g)(x) = \lim_{\varepsilon\to0}\int_{\Sigma\cap\{|x-y|>\varepsilon\}}\phi(x-y)g(y)ds(y),\quad C_\pm = \mp \frac{i}2(\alpha\cdot n) + C_s,
\end{equation}
where $n$ denotes the outward pointing normal to $\Omega_+$, see \cite[Lemma 3.3]{AMV14}.\\

\paragraph{\underline{Third step}}In order to define the operator one need to give a meaning to the expression $(u\delta_\Sigma)$ for $u = v + \Phi(g)\in E$. By analogy with \eqref{eqn:firstsymshell} for a spherical $\delta$-shell interaction, one can define this expression as the distribution
\begin{equation}\label{eqn:jumpcondgen}
	u\delta_\Sigma := \frac{1}2(u_+|_\Sigma + u_-|_\Sigma)\delta_\Sigma = \big(v|_\Sigma + C_s(g)\big)\delta_\Sigma.
\end{equation}
With this definition and using the jump formula, we compute $D_{\tau,\eta} u$  in the sense of distributions and obtain:
\begin{multline}
	D_{\tau,\eta} u = \bigg(- i(\alpha\cdot\nabla u_+) \mathds{1}_{\Omega_+}  - i(\alpha\cdot\nabla u_-)\mathds{1}_{\Omega_-}\bigg)\\ - i(\alpha\cdot n)(u_-|_\Sigma - u_+|_\Sigma)\delta_\Sigma + \frac12(\tau 1_4 + \eta\beta)(u_+|_\Sigma + u_-|_\Sigma)\delta_\Sigma.
\end{multline}
Thus, a natural jump condition through $\Sigma$ for $u\in E$ is
\begin{equation}\label{eqn:jumpcond}
	\frac12(\tau 1_4 + \eta\beta)(u_+|_\Sigma + u_-|_\Sigma) = i(\alpha\cdot n)(u_-|_\Sigma - u_+|_\Sigma).
\end{equation}
Taking the Plemelj-Sokhotski jump formula \eqref{eqn:PSjump} into account, it rewrites
\[
	(\tau 1_4 + \eta\beta) v|_\Sigma = -\big(1+(\tau+\eta\beta)C_s\big)g.
\]
It leads to the following definition of the operator $D_{\tau,\eta}$, that can be found in \cite[Thm. 3.8.]{AMV14} (for the pure electrostatic case $\eta=0$).
\begin{defi}\label{def:defamv} The operator $D_{\tau,\eta}$ is defined as
\[
	\left\{\begin{array}{rcl}
		\mathcal{D}(D_{\tau,\eta}) & = & \Big\{ v + \Phi(g) : v \in H^1(\mathbb{R}^3,\mathbb{C}^4), g\in L^2(\Sigma,\mathbb{C}^4),\\&&\quad\quad\quad (\tau 1_4 + \eta\beta) v|_\Sigma = -\big(1_4+(\tau1_4+\eta\beta)C_s\big)g\Big\},\\
		D_{\tau,\eta} (v + \Phi(g)) &=& (Dv_+) \oplus (Dv_-),\\
	\end{array}\right.
\]
where we have identified $L^2(\mathbb{R}^3,\CC^4)$ and $L^2(\Omega_+,\CC^4)\oplus L^2(\Omega_-,\CC^4)$.
\end{defi}
Contrary to the strategy developed for spherical shells (see \S \ref{subsec:DES}), Definition \ref{def:defamv} describes the functions in the domain of $\mathcal{D}(D_{\tau,\eta})$ as functions of the ambient Hilbert space $L^2(\mathbb{R}^3,\mathbb{C}^4)$ and precise their Sobolev regularity (actually, this can be made more precise as we will see thereafter in \S \ref{subsec:critstreng}).

The main result concerning self-adjointness reads as follows and concerns the pure electrostatic case ({\it i.e.} $\eta = 0$).

\begin{thm}[{\cite[Thm. 3.8.]{AMV14}}]\label{thm:amv14sa} Let $\Sigma$ be of class $C^2$. As long as $\tau \neq \pm 2$ the operator $D_{\tau,0}$ introduced in Definition \ref{def:defamv} is self-adjoint.
\end{thm}
Remark that we needed to impose two restrictions. The first one is that $\Sigma$ has to be sufficiently smooth but, the most surprising one, is the existence of two critical strengths for the coupling constants $\tau = \pm 2$. This last observation attracted a lot of attention in the past few years as we will see thereafter in \S \ref{subsec:critstreng}.

\begin{rem} An analogue of Theorem \ref{thm:amv14sa} has been obtained in \cite[Section 5.1]{AMV15} (more recently in \cite{BEHL19}) for the general operator $D_{\tau,\eta}$ and reads as follows.

Let $\Sigma$ be of class $C^2$. As long as $\tau^2 - \eta^2 \neq 4$ the operator $D_{\tau,\eta}$ defined in Definition \ref{def:defamv} is self-adjoint.
\end{rem}

Later on, Definition \ref{def:defamv} attracted the attention of specialists in self-adjoint extensions of symmetric operators acquainted with the theory of {\it quasi boundary triples}, a slight modification of the general theory of {\it boundary triples}. (see \cite{BGP07} and references therein for an introduction to boundary triples and \cite{BL07} for an introduction to quasi boundary triples). The main advantage of this theory is that it gives a systematic framework to define the operator, study its self-adjointness and spectral properties.

Following this path, in \cite{BEHL18}, the authors propose a definition of $D_{\tau,0}$  which coincides with the one given in Definition \ref{def:defamv} (see \cite[Definition 4.1.]{BEHL18}) and Theorem \ref{thm:amv14sa} is obtained as a consequence of the general theory of quasi boundary triples (see \cite[Thm. 4.4.]{BEHL18}).

The key argument in these two works lies in a link they establish between properties about the range and the kernel of an integral operator on the shell $\Sigma$ and the question of self-adjointness for $D_{\tau,0}$ (see \cite[Theorem 2.11]{AMV14} and \cite[Theorem 2.4]{BEHL18}).

It turns out that in this study, the anticommutator $\{C_s,i\alpha\cdot n\}$ plays a fundamental role. It is defined as
\begin{equation}\label{def:anticom}
	K := \{C_s,i\alpha\cdot n\} := i \big(C_s(\alpha\cdot n) + (\alpha\cdot n)C_s\big)
\end{equation}
and as long as the shell $\Sigma$ is of class $C^2$, $K$ is a compact operator from $L^2(\Sigma,\mathbb{C}^4)$ onto itself and the problem is solved by an adequate application of Fredholm alternative. Remark that the hypothesis on the smoothness of the shell $\Sigma$ plays a fundamental role here:  there is {\it a priori} no reason for this operator to be compact for less regular shells.

%%%%%%%%%%%%%%%%%%%%%
%%%%%%%%%%%%%%%%%%%%%
\subsection{How to handle the critical strengths ?}\label{subsec:critstreng}
%%%%%%%%%%%%%%%%%%%%%
%%%%%%%%%%%%%%%%%%%%%
The critical strengths that appear in Theorem \ref{thm:amv14sa} motivated the  simultaneous works \cite{BH17,OBV18} where the authors wonder until which extent the operator $D_{\tau,0}$ is self adjoint for the critical strengths $\tau =\pm 2$.

First, both works start with a different definition of the domain of the operator $D_{\tau,0}$.

\begin{defi} The Dirac operator with an electrostatic shell interaction of strength $\tau \in \mathbb{R}$ is denoted $D_\tau$ and defined as
\[
	\left\{\begin{array}{rcl}
		\cD(D_\tau) & = & \{ u = u_+ \oplus u_- \in H^1(\Omega_+,\CC^4) \oplus H^1(\Omega_-,\CC^4)) : \\
		&&\quad \quad \quad \quad \frac\tau2 (u_+|_\Sigma + u_-|_\Sigma) = i \alpha\cdot n (u_-|_\Sigma - u_+|_\Sigma)\},\\
		D_\tau u &=& (D u_+) \oplus (D u_-).
	\end{array}\right.
\]
\label{def:opcrit}
\end{defi}
The jump condition is the one obtained in \eqref{eqn:jumpcond} and the main result reads as follows (see \cite[Theorem 4.3]{OBV18} and \cite[Theorem 1.1. \& Theorem 1.2.]{BH17}).

\begin{thm} Let $\tau \in \RR$ and let $D_{\tau}$ be the operator of Definition \ref{def:opcrit}. The following alternative holds.
\begin{enumerate}[label=(\roman*)]
\item\label{itm:case1} If $\tau \neq \pm 2$, $D_\tau$ is self-adjoint and coincides with the operator $D_{\tau,0}$ of Definition \ref{def:defamv}.
\item\label{itm:case2} If $\tau = \pm 2$, $D_\tau$ is essentially self-adjoint and there holds
\[
	\begin{array}{rcl}
	\cD(D_\tau)\subsetneq \cD(\overline{D_\tau}) &:=& \{ u = u_+ \oplus u_- \in L^2(\RR^3,\CC^4)  : (\alpha\cdot\nabla)u_\pm \in L^2(\Omega_\pm,\CC^4),\\
		&&\quad \quad \quad \quad \frac\tau2 (u_+|_\Sigma + u_-|_\Sigma) = i \alpha\cdot n (u_-|_\Sigma - u_+|_\Sigma)\},
	\end{array}
\]
where the transmission condition holds in $H^{-\frac12}(\Sigma,\CC^4)$.
\end{enumerate}
\label{thm:sota}
\end{thm}

Because the difference of the resolvents of $D_\tau$ and $D_{\rm free}$ is a compact operator for the non-critical cases \ref{itm:case1}, an elemental spectral consequence of Theorem \ref{thm:sota} is that the essential spectrum of $D_{\tau}$ is given by
\[
	Sp_{\rm ess}(D_{\tau}) = Sp_{\rm ess}(D_{\rm free}) = \big(-\infty,-|m|\big] \cup \big[|m|,+\infty\big).
\]
Remark that this is not necessarily true for the critical cases, which may also prevent the functions in the domain $\cD(D_\tau)$ to have any Sobolev regularity. Namely, in \cite[Thm. 5.9.]{BH17} the authors prove the following theorem.
\begin{thm} Let $\tau = \pm 2$. If an open subset of $\Sigma$ is contained in a plane, there holds:
\[
	0 \in Sp_{ess}(D_\tau).
\]
In particular, for all $s>0$, $\cD(D_\tau)$ can not be included in the Sobolev space $H^s(\Omega_+,\CC^4)\oplus H^s(\Omega_-,\CC^4)$.
\label{thm:Sobregu}
\end{thm}

We briefly outline the strategy used to prove Theorem \ref{thm:sota} in \cite{OBV18}.\\

\underline{First step}. In order to prove Theorem \ref{thm:sota}, one needs to understand what is missing in the space $E$ in order to have an equality instead of an inclusion in \eqref{eqn:maxdomain}.

%Doing so is actually related to the mapping properties of the boundary integral operators introduced in \eqref{eqn:bdope} and \eqref{eqn:PSjump}.

To do so, remark that a duality argument implies that functions in $\cD(D_{\rm max})$ have weak traces in $H^{-\frac12}(\Sigma,\CC^4)$.

Then, one proves that the operator $\Phi$ introduced in \eqref{eqn:bdope} extends into a linear bounded operator from $H^{-\frac12}(\Sigma,\CC^4)$ to $\cD(D_{\rm max})$  (see \cite[Theorem 2.2.]{OBV18}).

Finally, remark that the operators $C_\pm$ of \eqref{eqn:PSjump} also extends as bounded operators from $H^{-\frac12}(\Sigma,\CC^4)$ onto itself.

The Plemelj-Sokhotski jump formula \eqref{eqn:PSjump} leads to introduce the bounded projectors in $H^{-\frac12}(\Sigma,\CC^4)$ defined as
\[
	\mathcal{C}_\pm := \pm C_\pm (i\alpha\cdot n).
\]
They satisfy $\mathcal{C}_+ + \mathcal{C}_- = Id$ and $\mathcal{C}_\pm^2 = \mathcal{C}_\pm$. They allow to describe accurately the maximal domain $\mathcal{D}(D_{\rm max})$. For this purpose, we are lead to introduce the spaces
\[
	\mathcal{H}_{0}^1(\Omega_\pm) := \{u \in H^1(\Omega_\pm,\CC^4): \mathcal{C}_\pm(u|_\Sigma) = 0\}.
\]
We have the following lemma.
\begin{lem}\label{lem:dirsum} The following direct sum of vector spaces holds.
\[
	H_\alpha(\Omega_\pm) = \mathcal{H}_{0}^1(\Omega_\pm)  \overset{\cdot}{+} \{ \Phi\big((\alpha\cdot n) f\big) : f\in \ran\mathcal{C}_\pm\},
\]
where $H_\alpha(\Omega_\pm) := \{u \in L^2(\Omega_\pm,\CC^4) : (\alpha\cdot\nabla) u \in L^2(\Omega_\pm,\CC^4)\}$.
\end{lem}
\begin{proof} It is clear that the set in the right-hand side is included in $H_\alpha(\Omega_\pm)$. Now, pick $u_\pm  \in H_\alpha(\Omega_\pm)$. We have
\[
	u_\pm = \underset{:= v_\pm}{\underbrace{u_\pm \mp i\Phi\big((\alpha\cdot n)\mathcal{C}_\pm (u_\pm|_\Sigma)\big)}}\pm i\Phi\big((\alpha\cdot n)\mathcal{C}_\pm (u_\pm|_\Sigma)\big).
\]
Remark that $v_\pm \in H_\alpha(\Omega_\pm)$ and  $v_\pm|_\Sigma = u_\pm|_\Sigma - \mathcal{C}_\pm(u_\pm) = \mathcal{C}_\mp(u_\pm|_\Sigma) \in H^{\frac12}(\Sigma,\CC^4)$ by \cite[Proposition 2.7.]{OBV18}. Thus, by elliptic regularity (see \cite[Proposition 2.16.]{OBV18}), $v_\pm \in H^1(\Omega_\pm,\CC^4)$ and as $v_\pm|_\Sigma = \mathcal{C_\mp}(u_\pm|_\Sigma)$ we get $C_\pm(v_\pm|_\Sigma) = 0$ and $v_\pm \in \mathcal{H}_0^1(\Omega_\pm)$. Setting $f = \pm i \mathcal{C}_\pm(u|_\Sigma)$ we obtain that
\[
H_\alpha(\Omega_\pm) =\mathcal{H}_{0}^1(\Omega_\pm)  {+} \{ \Phi\big((\alpha\cdot n) f\big) : f\in \ran\mathcal{C}_\pm\}
\]
It remains to prove that the sum is direct. Assume that
\[
u \in \mathcal{H}_{0}^1(\Omega_\pm) \cap \{ \Phi\big((\alpha\cdot n) f\big) : f\in \ran\mathcal{C}_\pm\}.
\]
Hence, $ u = \mp i \Phi\big((\alpha\cdot n)f\big)$ for some $f\in\ran\mathcal{C}_\pm$. As $u\in \mathcal{H}_{0}^1(\Omega_\pm)$, we obtain $0 = \mathcal{C}_\pm (u|_\Sigma) = f$ and $u = 0$.
\end{proof}

\begin{rem} The spaces $\{\Phi\big((\alpha\cdot n)f\big) : f\in \ran\mathcal{C}_\pm\}$ can be seen as analogues of Bergman spaces for the Dirac operator, similarly as the usual Bergman space defined as the space of square integrable holomorphic functions in a domain of $\mathbb{R}^2$. Moreover, remark that the space of traces of Dirac-harmonic functions in $\Omega_\pm$ is $\ran \mathcal{C}_\pm$. This space can be understood as the natural counterpart of Hardy spaces on the boundary $\Sigma$ for Dirac operators.
\end{rem}

As the maximal domain $\cD(D_{\rm max})$ satisfies
\[
	\cD(D_{\rm max}) = H_\alpha(\Omega_+) \oplus H_\alpha(\Omega_-),
\]
Lemma \ref{lem:dirsum} provides an accurate description of $\cD(D_{\rm max})$.\\

\underline{Second step}. Now, consider $D_\tau^*$, the adjoint of the operator $D_\tau$ introduced in Definition \ref{def:opcrit}. One can prove that
\[
	\cD(D_\tau^*) = \Big\{ u = u_+ \oplus u_- \in \cD(D_{\rm max}),
		\frac\tau2 (u_+|_\Sigma + u_-|_\Sigma) = i \alpha\cdot n (u_-|_\Sigma - u_+|_\Sigma)\Big\},
\]
where the transmission condition holds in $H^{-\frac12}(\Sigma,\CC^4)$. If the traces $u_\pm|_\Sigma \in H^{\frac12}(\Sigma,\CC^4)$, by elliptic regularity (see \cite[Proposition 2.16.]{OBV18}), the non-critical case \ref{itm:case1} Theorem \ref{thm:sota} is proved.

By Lemma \ref{lem:dirsum}, for $u \in \cD(D_{\rm max})$, we always have $\mathcal{C}_\pm(u_\mp|_\Sigma|) \in H^{\frac12}(\Sigma,\CC^4)$ and it remains to prove that $\mathcal{C}_\pm(u_\pm|_\Sigma|) \in H^{\frac12}(\Sigma,\CC^4)$. Using commutation relations between $\mathcal{C}_\pm$ and the multiplication operator $(\alpha\cdot n)$, one obtains the following system in $H^{-\frac12}(\Sigma,\CC^8)$ (see \cite[(4.8)]{OBV18}):
\begin{equation}\label{eqn:sysinv}
	A_\tau \begin{pmatrix}\mathcal{C}_+(u_-|_\Sigma)\\\mathcal{C}_-(u_+|_\Sigma)\end{pmatrix} = B_\tau \begin{pmatrix}\mathcal{C}_+(u_+|_\Sigma)\\\mathcal{C}_-(u_-|_\Sigma) \end{pmatrix} + F\begin{pmatrix} K(u_+|_\Sigma - u_-|_\Sigma) \\ K(u_+|_\Sigma - u_-|_\Sigma)\end{pmatrix},
\end{equation}
where $A_\tau,B_\tau, F \in C^1(\Sigma,\CC^{8\times8})$ and $K$ is the anticommutator introduced in \eqref{def:anticom}.

If $\Sigma$ is of class $C^2$ the right-hand side of \eqref{eqn:sysinv} belongs to $H^{\frac12}(\Sigma,\CC^8)$ because $K$ is not only compact as a bounded operator in  $L^2(\Sigma,\CC^4)$ but also a smoothing operator from $H^{-\frac12}(\Sigma,\CC^4)$ to $H^{\frac12}(\Sigma,\CC^4)$ (see \cite[Proposition 2.8.]{OBV18}). As $A_\tau$ is invertible if and only if $\tau \neq \pm 2$, \ref{itm:case1} Theorem \ref{thm:sota} is proved.

To prove \ref{itm:case2} Theorem \ref{thm:sota}, one first proves that $\overline{D_\tau}$, the closure of $D_\tau$, is $D_\tau^*$ and the only thing left to check is that $\cD(D_\tau)$ differs from $\cD(\overline{D_\tau})$. Actually, any $f \in H^{-\frac12}(\Sigma,\CC^4)$ such that $f\notin H^{\frac12}(\Sigma,\CC^4)$ generates a function in $\cD(\overline{D_\tau})$ which does not belong to $\cD(D_\tau)$ using that in this case, the matrix valued operator $A_\tau$ in \eqref{eqn:sysinv} is not invertible (see {\it e.g.} \cite[\S 4.4.]{OBV18}).
\subsection{State of the art on self-adjointness, consequences and open problems}\label{subsec:stofart}
In this paragraph, we give the most recent definition of the Dirac operator $D_{\eta,\tau}$ as given in  \cite{BEHL19}. This covers the previous definitions and results of \S \ref{subsec:genshell} and \S \ref{subsec:critstreng}.

\begin{defi}[{\cite[Equation (3.1)]{BEHL19}}]\label{def:final} Let $\tau,\eta \in \mathbb{R}$. The Dirac operator with electrostatic interaction of strength $\tau$ and Lorentz scalar interaction of strength $\eta$ denoted $D_{\eta,\tau}$ is defined as
\[\left\{
\begin{array}{rcl}
	\cD(D_{\tau,\eta}) &:=& \{ u = u_+ \oplus u_- \in H^{1}(\Omega_+,\CC^4)\oplus H^{1}(\Omega_-,\CC^4) : \\ &&\quad\quad i (\alpha\cdot n) (u_-|_\Sigma - u_+|_\Sigma) = \frac12\big(\tau 1_{4} +\eta \beta \big)(u_+|_\Sigma + u_-|_\Sigma)
	\},\\
	D_{\tau,\eta} u &=& (D u_+)\oplus(D u_-).
\end{array}\right.
\]
\end{defi}

Combining \cite[Theorem 3.4]{BEHL19}, \cite[Theorem 4.3.]{OBV18} and \cite[Theorems 1.1 \& 1.2]{BH17} we obtain the following result.
\begin{thm} If $\tau^2 - \eta^2 \neq 4$ the operator $D_{\tau,\eta}$ introduced in Definition \ref{def:final} is self-adjoint.

In the pure electrostatic case $\tau = \pm 2$ and $\eta=0$, $D_{\pm2,0}$ is essentially self adjoint and the domain of its closure $\overline{}$is given by
\begin{align*}
	\cD(D_{\pm 2,0})\subsetneq\cD(\overline{D_{\pm 2,0}}) &:=  \{u = u_+ \oplus u_- \in L^2(\mathbb{R}^3,\CC^4) : (\alpha\cdot\nabla)u_\pm \in L^2(\Omega_\pm,\CC^4)\\ &\quad\quad\quad i (\alpha\cdot n) (u_-|_\Sigma - u_+|_\Sigma) = \frac12\big(\tau 1_{4} +\eta \beta \big)(u_+|_\Sigma + u_-|_\Sigma)
	\},
\end{align*}
where the transmission condition makes sense in $H^{-\frac12}(\Sigma,\CC^4)$.
\label{thm:safinal}
\end{thm}

The missing part in this program is to understand until which extent the case $\tau^2 - \eta^2 = 4$ share the same features as the pure electrostatic critcal case $\tau = \pm 2$.

Let us finish this paragraph with a remark about the confinement of particles inside and outside the shell $\Sigma$. This condition is already stated in the seminal paper \cite[\S V.]{DES89} and developed in \cite[\S 5.]{AMV15}.
\begin{rem}\label{rmk:confinement} If $\eta^2 - \tau^2 = 4$ the shell generates confinement which physically means that it becomes impenetrable to the particles. This is nothing but a consequence of the fact that the traces of functions in the domain of $D_{\eta,\tau}$ are not coupled in this case and the operator can be rewritten as the direct sum of two operators with boundary conditions in $\Sigma$: one in $\Omega_+$ and another in $\Omega_-$ (see {\it e.g.} \cite[Lemma 3.1. (ii)]{BEHL19}. In the special case $\tau=0$ and $\eta =\pm 2$, one can easily observe this fact. For example, for $\tau = -2$, $D_{0,-2}$ can be rewritten as the direct sum
\[
	D_{0,-2} = D_+ \oplus D_-,
\]
where the self-adjoint operators $D_\pm$ are defined as
\[\left\{
\begin{array}{rcl}
	\cD(D_\pm) &:=& \{ u \in H^{1}(\Omega_\pm,\CC^4) :  \mathcal{B}_\pm u|_\Sigma = u|_\Sigma\},\\
	D_{\pm} u &=& D u,
\end{array}\right.
\]
where $\mathcal{B}_\pm := \mp i \beta (\alpha\cdot n)$. One recognizes the direct sum of two operators introduced in the late 60's and intensively studied in the physics literature (see the initial work \cite[\S IV]{Bo68} and the works \cite{Ch75,CJJT75,CJJTW74,Jo75}). These operators aim to model the confinement of quarks in hadrons and they are often referred to as {\it MIT bag operators} (see \cite{ALTR17} for a mathematical study).
\end{rem}

\subsection{Open problems}\label{subsec:saresult}
We conclude this section with some problems which are still open regarding the self-adjointness of the operator $D_{\tau,\eta}$ introduced in Definition \ref{def:final}.

\begin{oppb}\label{Oppb:sob}In the specific case $\tau^2 - \eta^2 = 4$ and for a smooth shell $\Sigma$, have the functions in the domain $\mathcal{D}(\overline{D_{m,\tau,\eta}})$ any Sobolev regularity ?
\end{oppb}
Open problem \ref{Oppb:sob} is partially answered in \cite[Thm. 5.9.]{BH17} where the authors consider a pure electrostatic shell interaction ($\eta=0$). When $\Sigma$ contains an open set included in a plane, a rather surprising spectral property appears: $0$ belongs to the essential spectrum of $\overline{D_{\tau,\eta}}$ which prevent the domain $\mathcal{D}(\overline{D_{\tau,\eta}})$ to be included in any Sobolev space $H^s(\Omega_+,\mathbb{C}^4) \oplus H^s(\Omega_-,\mathbb{C}^4)$ for all $s>0$. It is reminiscent of a similar phenomenon occurring in the study of metamaterials for which the geometry of the shell plays an crucial role to determine if whether or not some Sobolev regularity can be expected (see \cite{BDR99,BK18,CPP16}).

\begin{oppb}To our knowledge, all known results on self-adjointness of $D_{\tau,\eta}$ deal with sufficiently smooth shells $\Sigma$ (at least of class $C^2$). One may ask until what extent these results also hold for Lipschitz domains ?

In particular, the special class of corner geometries would deserve to be investigated. Indeed, it is known in the non-relativistic case that corners may generate interesting spectral features (see for instance \cite{DR14} for a broken line interaction). In this direction, let us mention the recent work \cite{PVDB19} in which the two-dimensional counterpart of $D_{0,\eta}$ is analyzed for the special case of a curves with finitely many corners.
\end{oppb}

\section{Approximation procedure}\label{sec:approxproc}
In this section we discuss the problem of the approximation of Dirac operators coupled with $\delta$-shell interactions using regularized hamiltonians. The main goal of this section is to recover Dirac operators coupled with $\delta$-shell interactions as limit of regularized hamiltonians as it is usually done in the non-relativistic setting (see \S \ref{sub:approx-Schr}).

The one dimensional version of this problem is tackled in \cite{Sebaklein} where it is proved that the regularizing sequence of hamiltonians converge in the norm resolvent sense to a Dirac point interaction: it answers question (Q3) in the relativistic setting.

However, when trying to answer question (Q4), a most surprising effect appears: the coupling constant in front of the point interaction depends non-linearly of the approximating sequence. This non-linear effect, understood as a reminiscence  of \emph{Klein's Paradox}, is a purely relativistic phenomenon due to the unboundedness neither from above nor below of the Dirac operator (see the original work of Klein \cite{Kle29} as well as \cite[\S 4.5]{Tha05} for a detailed explication of this phenomenom).

In dimension three, the problem is not entirely solved and this part of the review focus on the results presented in \cite{MPKlein} where a similar non-linear effect is exhibited. Following \cite{MPspherical}, we also discuss the special case $\Sigma = \mathbb{S}^2$, for which more properties can be deduced.

In \S \ref{subsec:mainresepsilon} we describe the current state of the art and \S \ref{subsec:squeezexpl} explains the main tools used to obtain these results. Finally, \S \ref{subsec:oppbepsilon} concludes this section with some problems which are still open about this approximation procedure.

\subsection{Main results}\label{subsec:mainresepsilon}

Let us start by defining the family of approximating potentials rigorously. In this section $\Omega$ is of class $C^2$ and for $\epsilon$ small enough, we introduce the tubular neighborhood of $\Sigma$ as
\[
\Omega_\epsilon := \{ x_\Sigma + t n(x_\Sigma) : x_\Sigma \in \Sigma, t\in (-\epsilon,\epsilon)\},
\]
where $n(x_\Sigma)$ is the normal to $\Sigma$ in $x_\Sigma$ pointing outward $\Omega_+$.

Fix $\epsilon_0>0$ sufficiently small in order for $\Omega_{\epsilon_0}$ to be in one-to-one correspondance with $\Sigma\times(-\epsilon_0,\epsilon_0)$. Consider $V \in L^\infty(\mathbb{R})$ with support in $[-\epsilon_0,\epsilon_0]$ from which we construct the family of squeezing potentials $(V_\epsilon)_{0<\epsilon<\epsilon_0}$ defined as
\begin{equation}\label{eqn:squeezpot}
	V_\epsilon(x) = \frac{\epsilon_0}{\epsilon} V\Big(\frac{\epsilon_0}{\epsilon}t\Big)\mathds{1}_{\Omega_\epsilon}(x),
\end{equation}
where for $x \in \Omega_\epsilon$, $t$ denotes its normal component in the decomposition $x = x_\Sigma + t n(x_\Sigma)$.

The main result of this section roughly reads as follows (see \cite[Theorem 1.2.]{MPKlein}).

\begin{thm} Let $V \in L^\infty(\mathbb{R})$ with compact support in $[-\epsilon_0,\epsilon_0]$ sufficiently small (which is precised further on). Let $(V_\epsilon)_{0<\epsilon<\epsilon_0}$ be the family of squeezing potentials constructed from $V$ as in \eqref{eqn:squeezpot}.

The following convergences hold in the strong resolvent sense.
\[
	D_{\rm free} + V_\epsilon 1_4 \underset{\epsilon\to0}{\longrightarrow} D_{\tau_V,0},\quad D_{\rm free} + V_\epsilon\beta \underset{\epsilon\to0}{\longrightarrow} D_{0,\eta_V},
\]
where the coupling constants $\tau_V,\eta_V \in \mathbb{R}$ depend non-linearly on the potential $V$.
\label{thm:cvgepsilon}
\end{thm}
Theorem \ref{thm:cvgepsilon} calls for numerous observations. First, the hypothesis of smallness on $V$ can be found, {\it e.g.} in \cite[Definition 1.1. \& Theorem 1.2.]{MPKlein}. It implies that $V$ is sufficiently small in $L^1(\mathbb{R})$.

Second, question (Q3) is answered and the authors prove a convergence which holds in the strong resolvent sense. It is not clear whether or not the norm resolvent convergence holds. In particular, the strategy used by the authors of \cite{MPKlein} suggests that a sequence of compact operators would converge to a non-compact operator, in the norm sense.

Finally, question (Q4) is explicitly answered. Namely, set
\begin{equation}\label{eq u,v}
\begin{split}
u(t):=|\epsilon_0 V (\epsilon_0 t)|^{1/2},\quad v(t):=\sign(V(\epsilon_0 t))u(t),\\ \mathcal{K}_V f(t):=\frac{i}{2}\int_{\RR}u(t)\sign(t-s)v(s)f(s)\,ds.
\end{split}
\end{equation}
$\mathcal{K}_V$ is an integral operator with kernel in $L^2(\mathbb{R}\times\mathbb{R})$, thus $\mathcal{K}_V$ is Hilbert-Schmidt and
\[
	\|\mathcal{K}_V\|_{HS} = \|u\|_{L^2(\mathbb{R})} \|v\|_{L^2(\mathbb{R})} = \|V\|_{L^1(\mathbb{R})}.
\]
Provided the operator $(1\pm\cK_V^2)$ is invertible, the coupling constants $\tau_V$ and $\eta_V$ are given by the following two formulas:
\begin{gather}\label{def lambda elec}
\tau_V:= \int_\RR v(t)\,((1-\cK_V^2)^{-1}u)(t)\,dt\in\RR, \\\label{def lambda scalar}
\eta_V:=\int_\RR v(t)\,((1+\cK_V^2)^{-1}u)(t)\,dt\in \RR.
\end{gather}
The hypothesis in Theorem \ref{thm:cvgepsilon} on the smallness of $V$ ensures that $(1\pm\cK_V^2)$ is invertible.

\begin{rem}\label{rem:confine} The hypothesis in Theorem \ref{thm:cvgepsilon} on the smallness of $V$ implies that $\|V\|_{L^1(\RR)} < 1$. Hence, let $\Lambda \in \{\tau_V,\eta_V\}$, we have:
\begin{align*}
|\Lambda| \leq \sum_{n\geq0} \int_{\RR} |v(t)| |\cK_V^{2n} u(t)| dt &\leq \|v\|_{L^2(\RR)}\|u\|_{L^2(\RR)} \sum_{n\geq 0} \|\cK_V\|_{HS}^{2n}\\
& = \frac{\|V\|_{L^1(\RR)}}{1 - \|V\|_{L^1(\RR)}}.
\end{align*}
As the smallness condition on $V$ in \cite{MPKlein} relies on the existence of $\delta := \delta(\Sigma) \in (0,\frac12)$ for which $\|V\|_{L^1(\RR)} \leq 2\delta < 1$, the previous equation implies that $|\Lambda|\leq\frac{2\delta}{1-2\delta}$ and the range of possible coupling constants is not the entire real line.

Actually, there is a simple example for which $\tau_V$ and $\eta_V$ can be computed explicitly. Let $\epsilon_0$ be as defined in \eqref{eqn:squeezpot} and $\delta= \delta(\Sigma) \in (0,\frac12)$. Set $V = \frac12\tau \mathds{1}_{(-\epsilon,\epsilon)}$ for $\epsilon \in(0,\epsilon_0)$ and $0<|\tau|\epsilon<2\delta_\Sigma$. One obtains
\[
\tau_V = 2 \tan\Big(\frac{1}2\tau\eta\delta(\Sigma)\Big),\quad \eta_V = 2 \tanh\Big(\frac12\tau\eta\delta(\Sigma)\Big).
\]
In particular $\tau_V, \eta_V \in(-2,2)$ and the endpoints $\pm 2$ can not be reached.

For the electrostatic coupling constant $\tau_V$, it corresponds to the critical coupling constants discussed in \S \ref{subsec:critstreng} for which the operator $D_{\tau,0}$ is essentially self-adjoint. Similarly, for the Lorentz-scalar coupling constant $\eta_V$ it corresponds to the confinement case described in Remark \ref{rmk:confinement}.
\end{rem}

\subsection{How to obtain the limit of the free Dirac operator coupled with squeezing potentials ?}\label{subsec:squeezexpl}

Let us tell a few words on the strategy developed in \cite{MPKlein} to prove Theorem \ref{thm:cvgepsilon}, focusing on the pure electrostatic case (the Lorentz scalar case being handled similarly).

The key idea is a resolvent formula for both the operators $D_{\rm free} + V_\epsilon$ and $D_{\tau_V,0}$. Indeed, for $z \in \CC\setminus \RR$ there holds:
\begin{equation}\label{eq:resolvent-free-delta}
\begin{split}
(D_{\tau_V}-z)^{-1}&=
(D_{\rm free}-z)^{-1}+
\tau_V A(z)(1+\tau_V B(z))^{-1}C(z),\\
%B_\tau(z),\\
(D_{\rm free}+\Vep-z)^{-1}&= 
(D_{\rm free}-z)^{-1}+
A_{V,\epsilon}(z)\left(1+ B_{V,\epsilon}(z)\right)^{-1}C_{V,\epsilon}(z).
%B_\tau(z),
%&=
%\Phi^z - A_\epsilon(z)\left(1+B_\epsilon(z)\right)^{-1}C_\epsilon(z).
\end{split}
\end{equation}
Here, $A(z), B(z), C(z)$ and $A_{V,\epsilon}(z),B_{V,\epsilon}(z), C_{V,\epsilon}(z)$ are bounded operators defined in \cite[Eqn. (2-13) \& Eqn. (3-5)]{MPKlein}. They are integral operators involving the fundamental solution of $D_{\rm free}-z$ introduced in the spirit of the one given in Proposition \ref{prop:def-fundsol}.

Thanks to \eqref{eq:resolvent-free-delta} it is enough to prove that in the strong sense
\begin{equation}\label{eq:limit-resolvent-bis}
A_{V,\epsilon}(z)\left(1+ B_{V,\epsilon}(z)\right)^{-1}C_{V,\epsilon}(z)
 \underset{\epsilon\to 0}{\longrightarrow}
\tau_V A(z)(1+\tau_V B(z))^{-1}C(z).
\end{equation}
Analyzing separately the convergence of these operators one obtains that
\begin{equation}\label{eq:convtoprove}
A_{V,\epsilon}(z)\underset{\epsilon\to 0}{\longrightarrow} A_{V,0}(z),\quad
B_{V,\epsilon}(z)\underset{\epsilon\to 0}{\longrightarrow} B_{V,0}(z)+B'_V,\quad
C_{V,\epsilon}(z)\underset{\epsilon\to 0}{\longrightarrow} C_{V,0}(z),
\end{equation}
where $A_{V,0}(z), B_{V,0}(z), B'_V(z)$ and $C_{V,0}(z)$ are defined in \cite[Eqn. (3-8)]{MPKlein}.

As the fundamental solution of $D_{\rm free} - z$ is not integrable near the origin (unlike for Schr\"odinger operators), the limiting operators may involve {\it singular integral operators} which calls for a thorough analysis. In particular $B_{V,\epsilon}$ converges in the strong resolvent sense to the sum of a {\it singular integral operator} $B_{V,0}$ and an integral operator $B'_V$. This is reminiscent of the Plemelj-Sokhostki formula \eqref{eqn:PSjump} (see \cite[Section 5]{MPKlein} for more details).

Now, the main difficulty is to prove that the operator $(1+B_{V,\epsilon})^{-1}$ converges in the strong sense to $(1+B_{V,0} + B_V')^{-1}$. To do so, one observes that
\[
	\|B_{V,\epsilon}(z)\|_{L^2(\RR^3,\CC^4) \to L^2(\RR^3,\CC^4)} \leq C(\Sigma) \|V\|_{L^1(\mathbb{R})},
\]
for some constant $C(\Sigma)>0$ depending on the shell $\Sigma$. Hence, if $\|V\|_{L^1(\mathbb{R})}$ is sufficiently small, $\|B_{V,\epsilon}(z)\|_{L^2(\RR^3,\CC^4) \to L^2(\RR^3,\CC^4)}< 1$ and one concludes that the following convergence holds:
\begin{equation}\label{eq:conv.ABCV}
A_{V,\epsilon}(1+B_{V,\epsilon})^{-1}C_{V,\epsilon}\underset{\epsilon\to 0}{\longrightarrow}
A_{V,0}(1+B_{V,0}+B'_V)^{-1}C_{V,0}  \quad \text{strongly}.
\end{equation}

The operator $B_V'$ is the one responsible for the non-linear dependance of $\tau_V$ on $V$. This operator does not appear in the analogous limit for the Schr\"odinger operator (see \cite[Section 3A]{MPKlein} for further details).

\subsection{Remarks and open problems}\label{subsec:oppbepsilon} For $\eta,\tau\in\RR$, the convergence of the operators
\[
D_{\rm free} + (\tau1_4 + \eta\beta)V_\epsilon,
\]
to a Dirac operator coupled with a shell interaction is probably the most significant physical justification for the investigation of these idealized models. Nevertheless, the actual state of the art answers only partially this question.

First, the smallness hypothesis in Theorem \ref{thm:cvgepsilon} is rather unsatisfactory because it is not well understood if it is a purely technical obstruction or if there is deeper physical signification for it to exist.
\begin{oppb}An answer to the following questions should shed some light on the problem of approximating the shell potentials for the Dirac operator.
\begin{enumerate}[label=(\roman*)]
\item Can one drop the smallness hypothesis on the potential $V$ in {\cite[Theorem 1.2.]{MPKlein}} ?
\item Is this restriction responsible for the range of possible $\eta_V,\tau_V$ to be different from the whole real line ?
\item If yes, is there another scaling of squeezing potentials which can give the remaining coupling constants ? In particular, does the obstruction to recover the confinement cases evoked in Remark \ref{rem:confine} can be overcome ?
\end{enumerate}
\end{oppb}

Second, the strong resolvent convergence does not provide all the spectral picture of the initial operators. Namely, any $\lambda \in Sp(D_{{\tau_V,0}})$ can be obtained as a limit when $\epsilon\to 0$ of some $\lambda_{\epsilon} \in Sp(D_{\rm free} + V_\epsilon)$ but the converse statement may not hold and spectral information can be lost in the limiting process.

In \cite{MPspherical}, a first attempt to study the converse proposition for the special case of a pure electrostatic potential for the spherically symmetric shell $\Sigma = \mathbb{S}^2$ is dealt with. The main strategy is to decompose in partial wave subspaces as explained in \S \ref{subsec:DES} and study the convergence for the one-dimensional radial operators obtained by this decomposition. In order to state the main result for this special case, consider the spherically symmetric potential
\[
	V_\epsilon = \frac{\mu}{2\epsilon}\mathds{1}_{(1-\epsilon,1+\epsilon)}(|x|).
\]
The operator $D_{\rm free} + V_\epsilon1_4$ decomposes similarly as in \S \ref{subsec:DES} in a countable family of differential operators in the radial variable $r$ as $d + \frac{\mu}{2\epsilon}\mathds{1}_{(1-\epsilon,1+\epsilon)}(r)$ acting in $L^2(\RR,\CC^2)$ as
\begin{align}\label{eqn:partialwaveapprox}
	\big(d + \frac{\mu}{2\epsilon}\mathds{1}_{(1-\epsilon,1+\epsilon)}(r)\big)u = -i\sigma_2 u' + m \sigma_3u + \frac{\chi}{r}\sigma_1u + \frac{\mu}{2\epsilon}\mathds{1}_{(1-\epsilon,1+\epsilon)}(r)u,\\ \nonumber\ \text{for some }\chi \in \ZZ\setminus\{0\}.
\end{align}
The main result of \cite{MPspherical} reads as follows.
\begin{thm}[{\cite[Theorem IV.2.]{MPspherical}}] Set $\tau = 2\tan\big(\frac\mu2\big)$.  Let {${\lambda_\epsilon \in Sp_{dis}\Big(d + \frac{\mu}{2\epsilon}\mathds{1}_{(1-\epsilon,1+\epsilon)}\Big)}$}, the discrete spectrum of a partial wave operator acting as in \eqref{eqn:partialwaveapprox}.\\ If {${\lambda_\epsilon \underset{\epsilon\to 0}{\longrightarrow} \lambda \in (-m,m)}$} then $\lambda\in Sp_{dis}(d_\tau)$.
\label{thm:convepsilonsphere}
\end{thm}

Theorem \ref{thm:convepsilonsphere} gives an insight on the origin of the eigenvalues in the gap of the operator $D_{\tau,0}$. However, nothing prevents the existence of a sequence $(\lambda_{\epsilon_n})_{n\in\mathbb{N}}$ accumulating at $\pm m$, which reads:
\[
\lambda_{\epsilon_n} \underset{\epsilon_n\to 0}{\longrightarrow}\pm m.
\]
The particular case of a spherical shell illustrates a phenomenon which may also hold for any generic shell $\Sigma$ of class $C^2$.
\begin{oppb} Do we have, for some $\gamma_\epsilon >0$ such that $\gamma_\epsilon \underset{\epsilon\to0}{\longrightarrow} 0$ 
\[
	\mathds{1}_{(-m+\gamma_\epsilon,m-\gamma_\epsilon)}(D_{\rm free} + V_\epsilon1_4) \underset{\epsilon\to0}{\longrightarrow}\mathds{1}_{(-m,m)}(D_{\tau_V,0}) \quad?
\]
Here, for $a\leq b \in \RR$, $\mathds{1}_{(a,b)}(D_{\rm free} + V_\epsilon1_4)$ and $\mathds{1}_{(a,b)}(D_{\tau_V,0})$ denote the spectral projectors on the line segment $(a,b)$ of $D_{\rm free} + V_\epsilon1_4$ and $D_{\tau_V,0}$, respectively. The convergence should hold in the topology of bounded operators in $L^2(\RR,\CC^4)$ and the compactness of $\Sigma$ should play a major role, as well as the behavior of $\gamma_\epsilon$ with respect to $\epsilon$.
\end{oppb}

\section{Structure of the spectrum}\label{sec:strucspec}
In this section we discuss properties of the spectrum of the Dirac operator with a shell interaction $D_{\tau,\eta}$ introduced in Definition \ref{def:final}. \S \ref{subsec:genprop} is about a resolvent formula and its consequences on the spectrum of the operator $D_{\tau,\eta}$.
\S \ref{subsec:spec-asympt} deals with a the large mass limit $m\to +\infty$ for the operator $D_{0,\eta}$.

\subsection{A resolvent formula and its consequences}\label{subsec:genprop}
The first work giving a resolvent formula for the operator $D_{\tau,\eta}$ is \cite{AMV15}, followed by the works \cite{BEHL18, BH17,BEHL19}. We focus on the non-critical cases $\tau^2 - \eta^2 \neq 4$. In \cite[Theorem 3.4.]{BEHL19} this resolvent formula writes as
\begin{equation}\label{eqn:resform}
	(D_{\tau,\eta} - z)^{-1} = (D_{\rm free} - z)^{- 1} - \Phi_z \bigg(\Big(1_4 + (\tau 1_4 + \eta \beta)C_{z,s}\Big)^{-1}\beta(\eta 1_4 + \tau \beta)\bigg) \Phi_z^*;\quad z \in \CC\setminus\RR,
\end{equation}
where the operators $\Phi_z$, $\Phi_z^*$ and $C_{z,s}$ are reminiscent of the one introduced in the \underline{Second step} of \S \ref{subsec:genshell}. Namely, if $\phi_z$ is a fundamental solution of $(D_{\rm free} - z)$, one defines the bounded operator from $L^2(\Sigma,\CC^4)\to L^2(\RR^3,\CC^4)$ by 
\[
	\Phi_z f(x) := \int_{s \in \Sigma} \phi_z(x - s) f(s) ds.
\]
$\Phi_z^*$ is its adjoint and $C_{z,s}$ is obtained thanks to a Plemelj-Sokhostki jump formula (see \eqref{eqn:PSjump}):
\[
	C_{z,\pm }(f) := \lim_{\Omega_\pm \ni y \overset{nt}{\to}x} \Phi_z(f) \in L^2(\Sigma,\CC^2),\quad C_{z,\pm }(f) := \mp \frac{i}{2}(\alpha\cdot n)f + C_{z,s}(f)
\]
and $C_{z,s}$ is a singular integral operator bounded from $L^2(\Sigma,\CC^4)$ onto itself. It is proved in \cite[Lemma 3.3.]{BEHL19} that $\big(1_4 + (\tau 1_4 + \eta \beta)C_{z,s}\big)$ is invertible. For the resolvent formula to be fully justified, one needs some extra mapping properties of these integral operators. This is the purpose of the following remark.

\begin{rem}
Actually, the mapping properties of these integral operators could be improved using Sobolev spaces on the boundary, following the same steps as \cite[\S 2.1. \& \S 2.2.]{OBV18}. Namely, these operators can be extended into bounded operators
\begin{multline*}
	\Phi_z : H^{-\frac12}(\Sigma,\CC^4) \to \cD(D_{\rm max}),\quad \Phi_z^* : L^2(\RR^3,\CC^4) \to H^{\frac12}(\Sigma,\CC^4),\\
	C_{z,s} : H^s(\Sigma,\CC^4) \to H^{s}(\Sigma,\CC^4)\quad\text{for}\quad s\in\{\pm\frac12\}.
\end{multline*}
Moreover, be elliptic regularity, if $f \in H^{\frac12}(\Sigma)$ then $\Phi_z(f) \in H^1(\Omega_+)\oplus H^1(\Omega_-)$.
\end{rem}

The main results concerning the structure of the spectrum, consequence of the resolvent formula \eqref{eqn:resform} can be recast as follows (see \cite[Theorem 4.1. \& Corollary 4.3.]{BEHL19}).

\begin{thm} Let $\tau^2 - \eta^2 \neq 4$. The following holds.
\begin{enumerate}
\item The essential spectrum of $D_{\tau,\eta}$ is
\[
	Sp_{ess}(D_{\tau,\eta})  = \big(-\infty,-|m|\big]\cup\big[|m|,+\infty).
\]
\item The discrete spectrum of $D_{\tau,\eta}$ is finite.
\item There exists $C>0$ such that the discrete spectrum of $D_{\tau,\eta}$ is empty if  $|\tau +\eta| < C$ or $|\tau - \eta| < C$.
\end{enumerate}
\end{thm}

Other consequences of the resolvent formula \eqref{eqn:resform} can be found {\it e.g.} in \cite{BEHL18}. See for instance the study of the non-relativistic limit \cite[Theorem 1.3.]{BEHL18} or the completeness of the wave operator for the scattering system $\{D_{\tau,0},D_{\rm free}\}$, see \cite[Theorem 1.2.]{BEHL18}.
\subsection{Spectral asymptotics}\label{subsec:spec-asympt}
In this paragraph we discuss the main results obtained in \cite{HOBP18} where the authors investigate the Dirac operator coupled with a Lorentz-scalar interaction of strength $\eta\in \RR$ denoted $D_{0,\eta}$ in Definition \ref{def:final}. In order to insist on the dependance of this operator on the mass parameter, in this paragraph, we set
\[
D_{\eta}(m) := D_{0,\eta}.
\]
By Theorem \ref{thm:safinal}, we know that this operator is self-adjoint. The following theorem can be found, {\it e.g.}, in \cite[ Theorem 2.3. \& Proposition 3.6.]{HOBP18}.

\begin{thm} Let $\eta \in \RR \setminus \{\pm 2\}$.
\begin{enumerate}[label=(\roman*)]
\item The essential spectrum of $D_{\eta}(m)$ is given by
\[
	Sp_{ess}\big(D_{\eta}(m)\big) = \big(-\infty,-|m|\big] \cup \big[|m|,+\infty\big).
\]
\item\label{itm:fin3} The spectrum of $D_{\eta}(m)$ is symmetric with respect to the origin.
\item The discrete spectrum of $D_{\eta}(m)$ is finite and each eigenvalue has an even multiplicity.
\item\label{itm:fin4}The operator $D_\eta(m)$ is unitarily equivalent to $D_{-\eta}(-m)$.
\item\label{itm:fin5} If $m\tau>0$, we have $Sp_{dis}\big(D_{\eta}(m)\big) = \emptyset$.
\end{enumerate}
\label{thm:specfeat1}
\end{thm}
\begin{rem} The condition $\eta \in \RR \setminus \{\pm 2\}$ is not a restriction. Indeed, as discussed in Remark \ref{rmk:confinement}, the operator uncouples as the direct sum of two Dirac operators with MIT bag boundary condition. In \cite{ALTR17}, the spectral properties of this model are investigated which provides a complete picture for all $\eta \in \RR$.
\end{rem}
By Theorem \ref{itm:fin5}-\ref{thm:specfeat1} the only interesting spectral feature that may happen in the gap $(-m,m)$ is when $m\tau < 0$ and by Theorem \ref{itm:fin4}-\ref{thm:specfeat1}, without loss of generality, we can pick $\tau < 0$ and $m >0$. The main result of \cite{HOBP18} reads as follows.

\begin{thm}\label{thm:HOBP} Let $\eta < 0$ with $\eta \neq 2$. The following Weyl-type asymptotics holds:
\[
	\# Sp_{dis}\big(D_\eta(m)\big) = \frac{16}{\pi} \frac{\tau^2}{(\tau^2 + 4)^2}|\Sigma|m^2 + \mathcal{O}\big(m\ln(m)\big),\quad m\to +\infty.
\]
Moreover, if $\pm \mu_k(m)$ denote the eigenvalues of $D_\eta(m)$ with $\mu_k(m) \geq 0$ enumerated in the non-decreasing order, then for each $k\in\NN$ there holds
\[
	\mu_k(m) = \frac{|\tau^2 - 4|}{\tau^2 +4}m + \frac{\tau^2 + 4}{|\tau^2-4|}\frac{E_k(\Upsilon_\tau)}{2m} + \mathcal{O}\Big(\frac{\ln(m)}{m^2}\Big),\quad m\to+\infty.
\]
Here, $E_k(\Upsilon_\tau)$ is the $k$-th eigenvalue of an m-independent Schr\"odinger operator $\Upsilon_\tau$ with an external Yang-Mills potential in $L^2(\Sigma,\CC^2)$
\[
	\Upsilon_\tau = \Big(d + i \frac{4}{\tau^2 +4}\omega\Big)^*\Big(d + i \frac{4}{\tau^2 +4}\omega\Big) - \Big(\frac{\tau^2 - 4}{\tau^2 + 4}\Big)^2 M^2 1_{2} + \frac{\tau^4 + 16}{(\tau^2 + 4)^2} K 1_2,
\]
where $K$ and $M$ are the Gauss and mean curvature, respectively. The $1$-form $\omega$ is given by the local expression
\[
	\omega := \sigma \cdot (n \times\partial_1n)ds_1 + \sigma\cdot(n\times\partial_2n)ds_2.
\]
\end{thm}
The Weyl asymptotics of Theorem \ref{thm:HOBP} justifies that the larger the mass $m$, the more eigenvalues are created in the gap $(-m,m)$. From a physical point of view, in the regime $m\to +\infty$, the system behaves at first order as the one for particles constrained to live on the shell $\Sigma$ driven by the effective hamiltonian $\Upsilon_\tau $. This hamiltonian is a geometric object as it is directly seen by the expression of the Yang-Mills potential. It is reminiscent of the recent work \cite{MOBP17}, where for the MIT bag Dirac operator (obtained by setting $\tau = \pm 2$, see Remark \ref{rmk:confinement}) the authors obtain an effective operator given by the square of the intrinsic Dirac operator on the shell $\Sigma$ (see {\it e.g.} \cite[Theorem 1]{MOBP17}). It leads to the following question.

\begin{oppb}
By analogy with \cite[Theorem 1]{MOBP17}, the effective operator $\Upsilon_\tau$ looks like the square of a Dirac operator with a {\it twisted} spin connection. Can its meaning be clarified ?
\end{oppb}

To conclude this paragraph, let us say a few words about the strategy used to prove Theorem \ref{thm:HOBP}. As the spectrum of $D_\eta(m)$ is symmetric with respect to the origin (see Theorem \ref{itm:fin3}-\ref{thm:specfeat1}), one can focus on the spectrum of the square $\big(D_\eta(m)\big)^2$. This is done by using a variational characterization of the eigenvalues of $\big(D_\eta(m)\big)^2$ {\it via} the min-max principle. Namely, for $u\in \cD(q_m) := \cD\big(D_\eta(m)\big)$, there holds:
\begin{align*}
	q_m(u) &:= \langle D_\eta(m) u, D_\eta(m) u\rangle_{L^2(\RR^3,\CC^4)}\\ &= \int_{\RR^3\setminus \Sigma} |\nabla u|^2 dx + m^2 \int_{\RR^3} |u|^2 dx\\
	&\quad\quad\quad\quad+ \frac{2m}\tau \int_{\Sigma}|u_+ - u_-|^2 ds + \int_{\Sigma}M|u_+|^2ds - \int_{\Sigma}M|u_-|^2 ds,
\end{align*}
where we have used the identification {${L^2(\RR^3,\CC^4)\ni u = u_+ \oplus u_- \in L^2(\Omega_+,\CC^4) \oplus L^2(\Omega_-,\CC^4)}$}. Remark that away from the shell $\Sigma$, this is the quadratic form of a shifted Laplacian and that the $\delta$-shell interaction manifests {\it via} the boundary terms.

Then, the strategy is rather standard, though technically involved because functions in the form domain $\cD(q_m)$ satisfy a specific transmission condition.
First, for $\delta >0$ sufficiently small one constructs a tubular neighborhood of the shell $\Sigma$ as
\[
	\Omega_\delta := \{ x_\Sigma - t n(x_\Sigma) : x_\Sigma \in \Sigma, t \in (-\delta,\delta)\}.
\]
Second, using Dirichlet and Neumann bracketting techniques, the unitary map
\begin{multline*}
	U : L^2(\Omega_\delta) \to L^2(\Sigma\times(-\delta,\delta)),\\ (Uu)\big(x_\Sigma - tn(x_\Sigma)\big) := \sqrt{\det\big(G(x_\Sigma,t)\big)}u\big(x_\Sigma - t n(x_\Sigma)\big) := v(x_\Sigma,t),
\end{multline*}
where $G$ is the metric on $\Sigma \times (-\delta,\delta)$ induced by the change of variable in tubular coordinates and remarking that the volume form on $\Sigma\times(-\delta,\delta)$ satisfies
\[
	\sqrt{\det \big(G(x_\Sigma,t)\big)} dx_\Sigma dt = \big(1-2tM(x_\Sigma) + t^2K(x_\Sigma)\big) ds dt,
\]
one can bound from above and below the quadratic form $\tilde{q}_m(v) := q_m(U^{-1}v)$ by quadratic forms $q_m^\pm$ defined on a form domain included in $U\cD(q_m)$. It reads
\[
	q_m^-(v) \leq \tilde{q}_m(v) \leq q_m^+(v).
\]
Hence, for the sequence of min-max levels associated to each quadratic form one gets (see \cite[Lemma 4.10]{HOBP18}):
\[
	E_k(q_m^-) \leq E_k(q_m) \leq E_k(q_m^+),\quad k \in \NN.
\]

The proof then relies on the obtention of a suitable upper bound for $q_m^+$ and a suitable lower bound for $q_m^-$. Let us focus on the simplest case, the upper bound.

For some constant $c>0$, we have:
\begin{align}
	q_m^+(v) &:= \int_{\Sigma\times(-\delta,\delta)} \Big((1+c\delta) \|\nabla_{x_\Sigma} u\|_{T_{x_\Sigma}\Sigma\otimes\CC^4} + (K-M^2+c\delta)|u|^2\Big) ds dt\nonumber\\&\quad\quad\quad + \int_{\Sigma}\Big(\int_{-\delta}^\delta |\partial_t u|^2 dt + \frac{2m}{\tau}|u(\cdot,0^+) - u(\cdot,0^-)|^2\Big)ds,
\label{align:fin}
\end{align}
where $v \in H^1(\Sigma\times(-\delta,\delta),\CC^4)$, satisfies a transmission condition at the shell $\Sigma$, inherited from the one for functions in the domain $\cD\big(D_\tau(m)\big)$ and a Dirichlet boundary conditions at $\{(x_\Sigma, \pm \delta) : x_\Sigma \in \Sigma\}$.

Now, picking a well chosen dependance of $\delta$ on the mass $m$, we set $\delta := \delta(m)$. One constructs trial functions as a product of the first mode of the transverse operator in the only variable $t$ which appears on the last line of \eqref{align:fin} and a function of the form $v_+(x_\Sigma) \mathds{1}_{t>0} + v_-(x_\Sigma)\mathds{1}_{t<0}$ where $v_+$ and $v_-$ are related in order for the trial functions to be in the form domain of $q_m^+$. Finally, one obtains the sought upper bound.

The lower bound is more involved because one has to control commutators of the surface gradient $\nabla_{x_\Sigma}$ with the projections on the eigenspace of the lowest eigenvalue of a transverse operator in the variable $t$ and its orthogonal, see {\it e.g.} \cite[\S 4.5.]{HOBP18} for details.
\section*{Acknowledgements}
T.O.-B. is supported by the ANR ``D\'efi des autres savoirs (DS10) 2017" programm,
reference ANR-17-CE29-0004, project molQED.

F. P. is supported by the European Research Council (ERC) under the European Union’s Horizon 2020 research and innovation program (grant agreement MDFT No 725528 of  Mathieu Lewin).

\end{document}